\documentclass[11pt]{amsart}

\usepackage{cite}
\usepackage[T1]{fontenc}
\usepackage[utf8]{inputenc}
\usepackage{lmodern}
\usepackage[a4paper,margin=38mm]{geometry}
\usepackage{microtype}
\usepackage{amsmath, amssymb, amsfonts,amsthm,amsopn,amscd, mathtools}
\usepackage{graphicx}
\usepackage{xcolor}
\usepackage[mathscr]{eucal}
\usepackage{thm-restate}
\usepackage{enumitem}
\usepackage[colorlinks=true,linkcolor=blue!85!black,citecolor=blue!85!black,urlcolor=blue!85!black]{hyperref}

\newtheorem{theorem}{Theorem}
\newtheorem{definition}[theorem]{Definition}
\newtheorem{lemma}[theorem]{Lemma}
\newtheorem{conjecture}[theorem]{Conjecture}

\newtheorem{corollary}[theorem]{Corollary}
\newtheorem{example}[theorem]{Example}
\newtheorem{question}[theorem]{Question}

\newcommand{\ot}{\otimes}
\newcommand{\Pauli}{\mathcal P}
\newcommand{\tr}{\operatorname{tr}}
\newcommand{\af}{\alpha}
\newcommand{\chif}{\chi_{\mathrm f}}

\newcommand{\AAA}{\mathcal{A}}
\newcommand{\cC}{\overline{C}_7}
\newcommand{\lex}{\operatorname{lex}}

\title[Fractional coloring for triply efficient shadow tomography]
{Counterexamples to the fractional coloring conjecture
for triply efficient shadow tomography}
\author{Jędrzej Stempin$^1$}
\author{Santiago Llorens$^1$}
\author{Felix Huber$^1$}
\date{\today}

\address{$^1$
Division of Quantum Information,
Institute of Informatics,
Faculty of Mathematics, Physics and Informatics,
University of Gdańsk,
Wita Stwosza 57, 80-308 Gdańsk, Poland
}
\email{felix.huber@ug.edu.pl}

\begin{document}

\begin{abstract}
Fractional graph colorings are useful for the Shadow tomography
of Pauli observables.
In practice, it is desirable that any experimentally interesting set of Pauli operators has a small 
fractional chromatic number $\chif$ for its anticommutation graph.
Conjecture 13 in King, Gosset, Kothari, and Babbush 
[PRX Quantum 6, 010336 (2025)]
states that if 
$B_\epsilon(\varrho)$ is the set of Pauli observables having expectation value
magnitude at least $\epsilon$ in some given quantum state $\varrho$, then the fractional
chromatic number of the  
anticommutation graph $G$ 
induced by $B_\epsilon(\varrho)$ is
$O(\epsilon^{-2})$.  
In other words, it asserts that
there exists a constant $C$ such that 
$\chif \cdot \epsilon^2 \leq C$
on all states and graphs.
If the conjecture were true, it would imply 
that there exists a 
triply efficient Pauli shadow tomography algorithm for {\it any} 
subset $S$ of Pauli observables,
provided that there is also an 
efficient fractional coloring algorithm for the set $B_\epsilon$.
Here we show that the conjecture is false by constructing a family of states and observables for which no finite $C$ 
satisfying the bound exists.
We also give a more general construction relying on the commutation index or $\beta$ number of a graph.
The key ingredient in the proofs can be seen as an 
instance of the amplification trick,
where fractional chromatic numbers, $\beta$ numbers, and expectation values are amplified through lexicographic graph products.

\end{abstract}

\thanks{
JS, SL, and FH were funded in whole or in part by the National Science Centre, Poland 2024/54/E/ST2/00451
and by the Polish National Agency for Academic Exchange under the Strategic Partnership Programme grant BNI/PST/2023/1/00013/U/00001.
For the purpose of Open Access,
the author has applied a CC-BY public copyright licence to any Author Accepted Manuscript (AAM) version arising from this submission.\\
AI statement: GPT Sol 5.6 was used to derive the main results (Theorem A and Theorem B) of this paper.
The authors verified and contextualized all results.}

\maketitle

\section{Introduction}

Let $\Pauli_n=\{I,X,Y,Z\}^{\otimes n}$
be the set of $n$-qubit Pauli strings.
The task of {\it Pauli shadow tomography} asks for simultaneous estimates of $\tr(\varrho P)$ for
the Pauli strings $P$ in a given set $ S\subseteq\Pauli_n$, using as few copies
of the unknown state $\varrho$ as possible.
A simple approach to tomography is to group the observables into disjoint sets of mutually commuting Pauli strings, 
which can then be  measured in common Clifford bases.  
Clearly, a larger number of disjoint sets requires more samples.
To understand more elaborate strategies, consider the anticommutation graph, whose vertices are connected if the corresponding Pauli strings anticommute.  Fractional coloring then allows one to measure using
distributions over commuting families instead of a single partition. 

More precisely, the task of shadow tomography is:
\begin{definition}[Shadow tomography~\cite{KGKB25}]
The shadow tomography task for a set $S$ of observables is as follows. We are given copies of an unknown state $\varrho$, and our goal is to output estimates $y_P$ such that (with high probability)
we have $|y_P - \tr(P\varrho)|\leq \epsilon$ for all $P \in S$.
\end{definition}

Let us describe two results by King, Gosset, Kothari, and Babbush~\cite{KGKB25}.
In both cases, one assumes that the fractional coloring of the corresponding anticommutation graph 
can be sampled with a classical randomized algorithm with runtime $T$.

The {\it single-copy} protocol learns a Pauli set $S \subseteq \Pauli_n$ with Clifford measurements with high probability within error $\epsilon$ 
with complexity~\cite[Theorem~5]{KGKB25}:
\begin{align}
    \text{samples:} \quad
        &O\big(\chif \log(|S|) / \epsilon^2\big)\,,\nonumber \\
    \text{runtime:} \quad
        &O\big((T + n^3) \chif (\log|S|) / \epsilon^2\big)\,.
\end{align}
Here $\chif$ is the fractional chromatic number of the  anticommutation graph of $S$.
Note that the lower bound on sample complexity is   
$\Omega\big(\chif / \epsilon^2\big)$ samples~\cite{Chen24}.

The {\it two-copy} protocol first uses Bell measurements to identify the set $S_\epsilon$ of Pauli strings with large expectation values, 
after which partitions arising from a fractional coloring of the anticommutation graph $G({S_\epsilon})$
are used to measure them. 
The complexity to learn $S$ to high probability with error $\epsilon$ is \cite[Lemma~20]{KGKB25}:
\begin{align}
\text{samples:}\quad 
    &O\big(\log(|S|)/\epsilon^4 + \chif\log(|S|)/\epsilon^2\big)\,,\nonumber\\
\text{runtime:}\quad 
    &O\big(|S|\log(|S|) / \epsilon^4 + (T+n^3)\chif \log(|S|) / \epsilon^2\big)\,.
\end{align}
In this case $\chif$ is the fractional chromatic number of the anticommutation graph induced by $S_\epsilon$ (not $S$).

One sees that for both the single- and two-copy protocols, 
the fractional chromatic number directly influences the sample and time complexities.

\bigskip
As one requires to learn Pauli expectations to within error $\epsilon$, define for any state $\varrho$ the set 
\begin{equation}\label{eq:B-def}
  B_\epsilon(\varrho)
  :=\{P\in\Pauli_n:|\tr(\varrho P)|\geq\epsilon\}.
\end{equation}

A key question is whether a large $\epsilon$ forces the anticommutation graph of $B_\epsilon(\varrho)$ to have a small fractional chromatic number~\footnote{In principle, one could distinguish the expectation-value threshold from the target accuracy $\epsilon$. Since the two are proportional in the shadow-tomography application, we use $\epsilon$ for both throughout.}. 
In particular, the following conjecture asserts that
$\chif(B_{\epsilon}(\varrho)) = O(1/\epsilon^2)$.
If true, the sample complexity of the two-copy protocol would collapse 
to $O(\log|S| / \epsilon^4)$, termed {\it sample-efficient}
in Ref.~\cite{KGKB25}~\footnote{In their work, any algorithm with sample complexity 
$\operatorname{poly}(\log|S|, \epsilon^{-1})$ is considered sample-efficient.}.

\bigskip
We state the conjecture by 
King, Gosset, Kothari, and Babbush verbatim~\footnote{Note that 
Ref.~\cite{KGKB25}
refers to commutation graphs, while we call the same objects anticommutation graphs.}.

\begin{conjecture}[Conjecture 13 in Ref.~\cite{KGKB25}]
\label{conj:13}
Let $\varrho$ be an $n$-qubit state, $\epsilon \in (0, 1)$, and let $B_{\epsilon} \subseteq \Pauli_n$ be the set of all
Paulis $P$ such that $| \tr(\varrho P )| \geq \epsilon$. 
There is a fractional coloring of the commutation graph
$G(B_\epsilon)$ of size $O(1/\epsilon^2)$.
\end{conjecture}

Rephrased as an inequality, Conjecture~\ref{conj:13} asserts:
there exists a constant $C$, such that for every state $\varrho$ and error $\epsilon$
\begin{equation}\label{eq:conjecture}
  \chif(B_\epsilon(\varrho)) \cdot \epsilon^{2} \leq C\,,
\end{equation}
where $\chif$ is the fractional chromatic number of the 
anticommutation graph of $B_\epsilon(\varrho)$.

In this paper we provide both a specific construction and a generic graph property that shows that such $C$ does not exist.

\subsection{Related works}

The general task of shadow tomography, introduced by Aaronson~\cite{aaronson_shadow_2018}, asks how many copies of an unknown state are required to estimate the outcome probabilities of \emph{any} set of two-outcome observables (i.e., not only Paulis). 
Among subsequent developments in the general shadow tomography problem~\cite{aaronson_gentle_2019,badescu_improved_2021}, recent work has obtained dimension-independent bounds~\cite{jeronimo_dimension-free_2026}.

A related but slightly different framework is that of classical shadows~\cite{huang_predicting_2020}. 
Here randomized measurements are used to construct a classical description of the state, 
from which one is able to reconstruct the expectation values of 
many (typically low-weight) observables.

Shadow tomography has also been studied in more structured settings, including thermal states with an accessible Hamiltonian~\cite{chen_efficient_2026} and Pauli and fermionic observables~\cite{KGKB25,Chen24}.

Note that our results concern the shadow tomography setting where one wants to measure a set of Pauli strings using few (i.e. a constant) number of copies.

\section{Contributions}
The aim of this paper is to show that Conjecture~\ref{conj:13} 
is false, equivalently, that an inequality as in Eq.~\eqref{eq:conjecture} cannot hold.

\begin{restatable}{thmA}{thmAinformal}\label{thm:thmA}
    There exists a family of states $\varrho_m$ and observables $\AAA^m$ violating Conjecture~\ref{conj:13} with $\chif(B_{\epsilon_m})$ 
    scaling as
    $\Omega(\epsilon^{-2.07598})$. 
\end{restatable}
This is proven in Theorem~\ref{thm:A_details} and 
Corollary~\ref{cor:scaling}
in Section~\ref{Sec:proof_main}. 
It relies on the amplification~\cite{Tao07} of certain graph properties through the lexicographic product.

\bigskip
A more general statement can be made using the machinery of $\beta$ numbers, a graph invariant that is related to quantum properties.
It is compared against the independence number $\alpha$,
that is the size of the largest subset of disconnected vertices.

\begin{restatable}{thmA}{thmBinformal}
    \label{thm:thmB}
Let $G$ be a graph with $\alpha(G) < \beta(G)$.
Then $G$ provides a counterexample to Conjecture~\ref{conj:13}.
\end{restatable}

In particular, graphs with this property exist. For example, 
$\overline C_7$ satisfies the hypothesis $\alpha < \beta$ of Theorem~B, so the counterexample of Section~\ref{Sec:proof_main} can be seen as an instance of Theorem~\ref{thm:general-amplification} but with a simpler proof.

Theorem~\ref{thm:thmB} is proven via Theorem~\ref{thm:general-amplification} in Section~\ref{sec:beta-number-amplification}.
It also relies on an amplification trick,
together with an inequality for partial sums.

\section{Proof sketch}

We sketch the proof of Theorem~\ref{thm:thmA}. 
Our second main result, Theorem~\ref{thm:thmB}, 
has a similar flavour; 
however, it requires the introduction of slightly more machinery.
In what follows we drop $\varrho$ from $B_\delta(\varrho)$ 
when the state in question is clear, and likewise for other sets defined by a state.

We provide a counterexample that violates Inequality~\eqref{eq:conjecture} via the following construction: 
Consider the anti-heptagon $\cC$,
and denote by $\AAA = \{A_1, \dots, A_7\}$ seven Pauli strings which realize it as an anticommutation graph.  
Consider $H = \sum_{A \in \AAA} A$ and denote its ground state by $\varrho$.
One can check that every operator $A \in \AAA$ 
has the same expectation value, namely
\begin{equation}
  |\tr(\varrho A)|=a\,, \quad \text{for all} \quad A \in \AAA\,, \qquad \text{where} 
  \qquad
  a= \frac{1+2\sqrt2}{7}\,.
\end{equation}
Since the fractional chromatic number of the anti-heptagon is $\chif(\cC)=7/2$, 
it follows that
\begin{equation}\label{eq:geqone}
  \chif(\cC) \cdot a^2
  =\frac{9+4\sqrt2}{14} \approx 1.046918 >1.
\end{equation}

The key idea is to amplify this inequality through 
lexicographic graph products (illustrated in Fig.~\ref{fig:lex}), 
so that the right hand side of Eq.~\eqref{eq:geqone} becomes unbounded.
To this end, take the $m$-fold product of~$\cC$, yielding the graph 
$\cC^{ \lex m}$.
Corresponding to it are a state $\varrho_m$ and a set $\AAA^{m}$ containing $7^m$ Pauli observables,
which are constructed in more detail in Section~\ref{Sec:proof_main}. 

The amplification trick uses that 
\begin{align}
     \chif(G^{\operatorname{lex} m}) &= \chif(G)^m\,, \nonumber\\
     |\tr\big(\varrho_m A\big)| &= a^m 
     \quad\quad \text{for all}\quad A \in \AAA^{m}\,.
\end{align}
The first property follows from the well-known fact $\chif(G \operatorname{lex} H) = \chif(G) \cdot \chif(H)$ \cite[Corollary~3.4.5]{SU11}. We show the second property in Section~\ref{Sec:proof_main}.

Now write $G(S)$ for the anticommutation graph induced 
by the set $S$.
Also, from now on use a shorthand $\chif(S) = \chif(G(S))$ for any observable set $S$.
Set $\epsilon_m=a^m$. 
Note that 
for the state $\varrho_m$,
$G(\AAA^{m})$ is an induced subgraph of  $G(B_{\epsilon_m})$.

The fractional chromatic number is monotone under induced
subgraphs, and $\chif(\AAA^{m}) \leq \chif(B_{\epsilon_m})$ holds.

This yields the following inequality
\begin{equation}
    \chif(B_{\epsilon_m}) \cdot \epsilon_m^2 
    \quad\geq  \quad
    \chif(\AAA^m) \cdot \epsilon_m^2 
    \quad = \quad
    \left(\frac{9+4\sqrt2}{14}\right)^m \approx (1.046918)^m \,.
\end{equation}

The right hand side is unbounded under repeated lexicographic graph products, i.e. when $m \longrightarrow \infty$.
Hence the inequality~\eqref{eq:conjecture} cannot hold, and Conjecture~\ref{conj:13} is false.

\section{Concepts}

\subsection{Anticommutation graphs}

Consider a set $S\subseteq\Pauli_n$ of Pauli strings. 
Its \emph{anticommutation graph} $G(S)$ has as vertices the elements of $S$, with edges between them if the corresponding Pauli strings anticommute, 
\begin{equation}
    i \sim j \quad \text{if} \quad A_i A_j = - A_j A_i\,.
\end{equation}
Naturally, an independent set in $G(S)$ corresponds to a pairwise commuting family of Paulis. 

We call a set of ($\pm$ signed) Pauli strings $S$ a {\it realization} of graph $G$, if the anticommutation graph of $S$ equals $G$.
Every graph $G$ allows for a realization
in terms of Pauli strings \cite{gastineau1982quasi}.

\subsection{Independence, chromatic, and $\beta$ numbers}
\label{sec:chromatic_number}
Denote the maximum size of an independent set, 
the independence number, as $\af(G)$.
Let $\mathcal I(G)$ denote the set of independent sets of a graph $G$.
Denote its independence number by
\begin{equation}
    \alpha(G)=\max_{I\in\mathcal I(G)} |I|\,.
\end{equation}
The \emph{chromatic number} $\chi(G)$ is the minimum number of colors that can be assigned to the vertices such that no two adjacent vertices share a color.

The \emph{fractional chromatic number} is a linear-programming relaxation of the chromatic number,
\begin{equation}\label{eq:chif-def}
 \chif(G) =\min\left\{
   \sum_{I\in\mathcal I(G)}w_I:
   w_I\geq0,\quad
   \sum_{I \, : \, v\in I}w_I\geq1\ \text{for every }v\in V(G)
 \right\}.
\end{equation}
Here each independent set receives a nonnegative weight and every vertex must receive total weight at least one.

Finally, for $\AAA$ a realization of $G$, define the
$\beta$ number as
\begin{equation}
\label{eq:beta-number}
    \beta(G) =\max_{\varrho}\sum_{A\in\AAA}
    \tr(\varrho A)^2\,.
\end{equation}
In particular, $\beta$ is independent of the particular Pauli string realization~\cite[Section~IV,~Theorem~3]{XSW24}.

We will use the following four facts. These are all standard (see e.g.~\cite{godsil2001algebraic, SU11}) 
apart from the monotonicity and multiplicativity of $\beta$~\cite{XSW24}
and Item~\ref{item:fact-d}~\cite{hastings2022optimizing}.

\begin{enumerate}[label=(\alph*), itemsep = 0.4em]
    \item\label{item:fact-a}
    Let $H$ be an induced
subgraph of $G$. Then
\begin{equation}
\label{eq:chi_monotone}
 \alpha(H)\leq\alpha(G)\,, \quad \chif(H)\leq\chif(G)\,, \quad \beta(H)\leq\beta(G)\,.
\end{equation}

\item\label{item:fact-b} $\alpha$, $\chif$, $\beta$ are multiplicative under the lexicographic product (cf. the next section),
\begin{align}
\label{eq:lex_product_multi}
    \alpha(G \lex H) &= \alpha(G) \cdot \alpha(H)\,,\nonumber \\
    \chif(G \lex H) &= \chif(G) \cdot \chif(H)\,,\nonumber \\
    \beta(G \lex H) &= \beta(G) \cdot \beta(H)\,.
\end{align}
\item\label{item:fact-c} For all graphs,
\begin{equation}
\label{eq:vertex-transitive-chif}
    \chif(G) \geq \frac{|V(G)|}{\alpha(G)} \,.
\end{equation}
Equality holds if $G$ is vertex-transitive.

\item\label{item:fact-d} For all graphs, $\alpha(G) \leq \beta(G)$.

\end{enumerate}
Hereafter, we simply use $|V|$ for $|V(G)|$ when the context is clear.

\subsection{Lexicographic graph products}
\begin{figure}[tbp]
    \includegraphics[width = 0.6\textwidth]{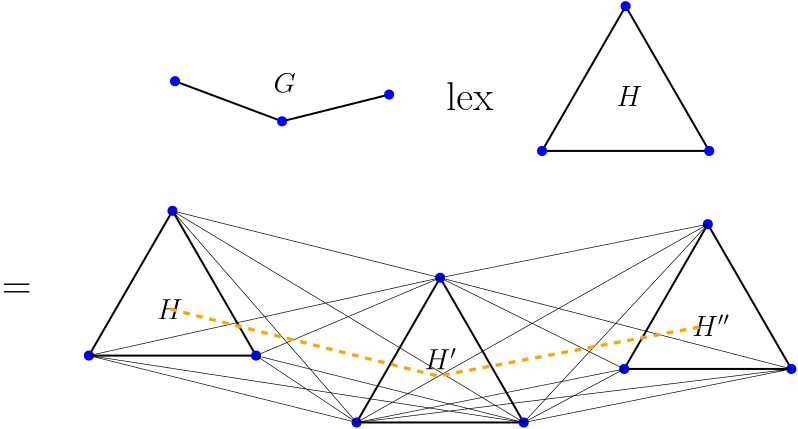}
    \caption{The lexicographic product  of the line graph with the triangle, $L_3 \lex C_3$.
    At every vertex of $G$ (a branch), 
    place a copy of $H$ (a leaf).
    Then connect the vertices 
    among different leaves 
    if and only if their branches are connected.
    \label{fig:lex}}
\end{figure}

Let $G = (V,E)$ and $H = (W,F)$. 
The \emph{lexicographic product} $G \lex H$ has vertex set
$V\times W$ and edges
\begin{equation}\label{eq:lex-adjacency}
  (i,j)\sim(k,l)
  \quad\quad \text{if} \quad\quad
  \begin{cases}
      i \sim k\ \text{or}\\
      i=k \text{ and } j \sim l
  \end{cases}
\end{equation}
The idea is simple: at each vertex of $G$ (branch) 
place a copy of $H$ (a leaf). 
Then connect the vertices among different leaves if and only if their branches are connected.
This is illustrated in Fig.~\ref{fig:lex}.
We write $G^{\lex 1} = G$ and $G^{\lex m+1} = (G^{\lex m} \lex G)$ for the lexicographic powers.

\subsection{Realizing observables for lex products}\label{sec:lex_prod_obs}
Given two observable sets $S$ and $T$ 
with anticommutativity graphs $G$ and $H$,
how can the $G \lex H$ be realized?
We give a small example:

\begin{example}
\label{lexicographic_example}
Let $G = L_3$ and $H = C_3$. 
The graph $G$ is realized by $T = \{XI, ZZ, IX\}$,
and 
$H$ is realized by $R = \{X,Y,Z\}$.
A set realizing $(G \lex H)$ is
\begin{align}
    \text{leaf $H$:}   \quad\quad &\mathrm{XI\,|\,XII}\,,\quad \mathrm{XI\,|\,YII}\,,\quad  \mathrm{XI\,|\,ZII}\,, \nonumber\\ 
    \text{leaf $H'$:}  \quad\quad &\mathrm{ZZ\,|\,IXI}\,,\quad \mathrm{ZZ\,|\,IYI}\,,\quad  \mathrm{ZZ\,|\,IZI}\,,\nonumber\\
    \text{leaf $H''$:} \quad\quad &\mathrm{IX\,|\,IIX}\,,\quad \mathrm{IX\,|\,IIY}\,,\quad  \mathrm{IX\,|\,IIZ}\,.
\end{align}
Here the first register is occupied by the observables $T$ of $G$.
The second register hosts the observables $R$ of $H$, 
with a separate tensor factor (three of them in total) 
reserved for each copy of $H$.
\end{example}

We now follow with a generic construction~\cite{XSW24}. 
Let $T = \{T_1, \dots, T_t\}$ and $R = \{R_1, \dots, R_r\}$ realize $G$ and $H$ respectively.
Then the following set realizes $G \lex H$ as anticommutation graph,
\begin{equation}
\label{eq:lex_operators}
    T \lex R \,\,=\,\, \Big\{ 
    \overbracket{T_i}^{\text{Reg. $1$}} 
    \ot 
    \,\,
    \overbracket{
    I \ot \cdots \ot I \ot \hspace{-0.7cm}
    \underbracket{R_j}_{\text{$i$'th tensor factor}} 
    \hspace{-0.6cm} \ot\, I\ot  \cdots \ot I
    }^{\text{Reg. $2$}}
    \,\,\Big|\, i=1,\dots, t\,;\,\, j=1, \dots, r\Big\}\,.
\end{equation}
Here $T_i$ is placed on the first register, 
and $R_j$ placed onto the $i$-th tensor factor of the second register,
where the second register contains $|V(G)|$ internal tensor factors.
In the following we denote by $T^{\lex m} = (T \lex \dots \lex T)$
($m$ times) the $m$-fold lex product of the observable set~$T$.

\section{Anti-heptagon as a counterexample} \label{Sec:proof_main}

We now construct the set of states and Pauli strings that serve as counterexample to Conjecture~\ref{conj:13}.
Consider the following set that realizes the anti-heptagon~$\cC$
(see Fig.~\ref{fig:c7-graphs}),
\begin{align}
\label{eq:pauli_set1}
\mathcal{A} = \big\{
\mathrm{ZZI}, 
\mathrm{ZII}, 
\mathrm{IXI}, 
\mathrm{XII}, 
\mathrm{XZX}, 
\mathrm{YZZ}, 
\mathrm{YYY}
\big\}\,.
\end{align}

\begin{figure}[tbp]
\centering
\includegraphics{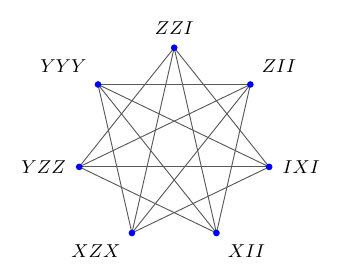}
\caption{Anticommutation graph of $\cC$ and its Pauli realization.
\label{fig:c7-graphs}}
\end{figure}

Consider the Hamiltonian  
$H =  \sum_{A \in \mathcal{A}}A$
and let $\varrho_{\min}$ be its ground state (see Appendix~\ref{app:seed}).
One can check that the expectation values are equal for all observables $A\in \AAA$
\begin{equation}\label{eq:value_a}
    |\tr(A\varrho_{\min})| = \frac{(1 + 2\sqrt{2})}{7} =: a\,.
\end{equation}

We also note that $\chif(\cC) = \tfrac{7}{2}$, which follows from the fact that $\cC$ is vertex transitive and $\alpha(\cC) = 2$.

\subsection{Amplified observables and state}
\label{sec:amplification_constr}
As amplified set of observables, we 
take the observables that realize the $m$-fold lex product $\cC^{\,\lex m}$,
namely
\begin{align}\label{eq:amp_obs}
    \AAA^m = \AAA^{\lex m} = (\AAA \lex \dots \lex \AAA) \quad (m \text{ times})\,.
\end{align}
As amplified state, we take 
\begin{align}\label{eq:amp_state}
    \varrho_m = \varrho_{\min}^{\ot K}\,,
\end{align}
where $K$ is the total number of tensor factors required for the construction of $\AAA^m$.
Due to the fact that on each register there is only one non-trivial tensor factor in the state, it follows that
\begin{align}
    |\tr(\varrho_m A)| = a^m \quad \text{for all}\quad A \in \AAA^m\,.
\end{align}

\subsection{Counterexample}
\label{sec:counterexample}
We are ready to state the counterexample.

\begin{theorem}\label{thm:A_details}
The anti-heptagon provides a family of states $\varrho_m$ and observables $\AAA^m$ violating Conjecture~\ref{conj:13}. 
\end{theorem}
\begin{proof}
Let $\AAA^m$ and $\varrho_m$ be as given in Eq.~\eqref{eq:amp_obs} and \eqref{eq:amp_state}.
Then consider the set
\begin{equation}
  B_{\epsilon_m}(\varrho_m)
  = \big\{P \in \Pauli_n\,:\, |\tr(\varrho_m P)|\geq \epsilon_{m} \big\},
\end{equation}
where $n$ is chosen large enough to host the state $\varrho_m$. From now on write $B_{\epsilon_m}$ for $B_{\epsilon_m}(\varrho_m)$.

Set $\epsilon_m = a^m$.
Clearly, the set $\AAA^m$ is a subset of $B_{\epsilon_m}$,
because every element in $\AAA^m$ has 
expectation value magnitude of exactly $\epsilon_m$.
It follows that  $G(\AAA^m)$ is an induced subgraph of $G(B_{\epsilon_m})$. 
By Fact~\ref{item:fact-a} on fractional chromatic numbers [Eq.~\eqref{eq:chi_monotone}],
one has $\chif(\AAA^m)  \leq \chif(B_{\epsilon_m}) $.
Then
\begin{align}
\label{eq:main_inequality_1}
    \chif(B_{\epsilon_m}) \cdot \epsilon_m^2
    \quad&\geq \quad
    \chif(\AAA^m) \cdot a^{2m} \nonumber\\
    &=\quad \Big(\frac{9+4\sqrt2}{14}\Big)^m \quad \approx\quad  (1.046918)^m \,.
\end{align}
Here we have used the value of $a = \frac{(1 + 2\sqrt{2})}{7}$ from Eq.~\eqref{eq:value_a} and the multiplicativity of $\chif$ [Eq.~\eqref{eq:lex_product_multi}].
Letting $m \longrightarrow \infty$, one can see that the right hand sied in above inequality is unbounded.

Thus there does not exist a finite constant $C$ such that $\chif(B_\epsilon(\varrho)) \epsilon^2 \leq C$ holds for all graphs, states, and errors. 

This ends the proof.
\end{proof}

\subsection{Scaling}
Let us now consider the scaling of the fractional chromatic number in the case of the anti-heptagon.

\begin{corollary}\label{cor:scaling}
    For lexicographic powers of the anti-heptagon, the fractional chromatic number of $B_{\epsilon_m}(\varrho_m)$ scales as 
    $ \Omega({\epsilon_m}^{-2.07598})$.
\end{corollary}

\begin{proof}
Note that the proof in Theorem~\ref{thm:A_details}
still provides a counterexample to 
$\chif$ scaling as $O(\epsilon^{-s})$, if
\begin{equation}
\chif(\AAA) a^s > 1\,,
\end{equation}
 due to the amplification occurring in Eq.~\eqref{eq:main_inequality_1}.
Taking a logarithm of the condition,
\begin{equation}
    \log(\chif(\AAA)) + s \log(a) > 0\,.
\end{equation}
Thus
\begin{equation}
    s < - 
    \frac{\log(\chif(\AAA))}
    {\log(a)}
    =
    - \frac{\log(\tfrac{7}{2})}
    {\log(a)} 
    \approx 2.07598\,.
\end{equation}
Thus necessarily, 
$\chif(B_{\epsilon_m})$ must scale at least as
$\Omega(\epsilon^{-2.07598})$.
\end{proof}

\section{Amplification through the $\beta$ number}
\label{sec:beta-number-amplification}

Let $G$ be realized as an anticommutation graph by a set of Pauli strings $\AAA$. Consider the $\beta$-number (or commutativity index),
\begin{equation}\label{eq:beta-def2}
 \beta(G) = \max_{\varrho}
 \sum_{A \in \AAA}\tr(\varrho A)^2\,.
\end{equation}
Here the maximization is taken over all quantum states $\varrho$.
The $\beta$ number
was introduced in Ref.~\cite{XSW24,hastings2022optimizing}
and since then appeared in many different contexts in quantum information theory \cite{ de2023uncertainty, moran2024uncertainty, Chen2022Exponential, XSW24, anschuetz2025strongly, Xu2025Simultaneous, munne2024sdp}.

Denote by $\alpha$ and $\vartheta$ the independence number and the Lovász number of a graph. It holds that~\cite{de2023uncertainty,hastings2022optimizing},
\begin{equation}\label{eq:sandwich_thm}
    \alpha(G) \leq \beta(G) \leq \vartheta(G)\,.
\end{equation}
Ref.~\cite{XSW24} showed that the left-hand side of the 
Eq.~\eqref{eq:sandwich_thm} is {\it not} tight. 
In particular, the anti-heptagon has $2 = \alpha(\cC) < \tfrac{9 + 4\sqrt{2}}{7} \leq \beta(\cC)$.

Now note that the proof in 
Section~\ref{Sec:proof_main}
uses the fact that the expectation values of all $A\in \AAA$ on $\varrho$ 
are both constant and sufficiently high. 
Thus it is tempting to think that a high $\beta$ number
alone will be sufficient for constructing a counterexample to Conjecture~\ref{conj:13}.

This is indeed the case.

\begin{restatable}{theorem}{theorem_beta}
[Amplification of $\beta$]\label{thm:general-amplification}
Let $G$ be a graph on $N>0$ vertices.
Then there exists a sequence of states $\sigma_{m}$ and
$\epsilon_{m}>0$, such that 
\begin{equation}\label{eq:bound1} 
    \chif(B_{\epsilon_{m}})\cdot \epsilon_{m}^{2} \quad \geq \quad \frac{0.99}{1+m\log N} \bigg( \frac{\beta(G)}{\alpha(G)} \bigg)^{m}\,.
\end{equation}
\end{restatable}

Note that if $\alpha(G)<\beta(G)$, 
then the expression on the right-hand side of the Eq.~\eqref{eq:bound1} is unbounded as $m\rightarrow \infty$:
the denominator $1+ m\log N $ grows linearly while $\big(\tfrac{\beta(G)}{\alpha(G)}\big)^m$ grows exponentially. 
Thus, there does not exist a constant $C$ for which $\chif \cdot \epsilon^{2}\leq C$. 

\bigskip
In particular, graphs with $\alpha(G)<\beta(G)$ exist~\cite{XSW24}.
It follows that:

\thmBinformal*

We state the proof of Theorem~\ref{thm:general-amplification} 
in Section~\ref{subs:beta_theorem_proof}
and for now proceed with a rather standard auxiliary Lemma.

\subsection{Threshold Lemma}
Given a high value of $\beta$, 
we need to select those observables in $\AAA$ that show a sufficiently high squared expectation value.
While the following proof is standard
we reproduce it here for the reader.
For this denote by $h_{N}$ 
a partial sum of the harmonic series, 
\begin{equation}\label{eq:part_harmonic}
    h_{N} = \sum_{k=1}^{N}\frac{1}{k}\,.
\end{equation}

\begin{lemma}
\label{lem:threshold-selection}
Let $Y = (y_i)_{i=1}^N$ be a  non-increasing sequence of non-negative numbers, not all equal to zero,
\begin{equation}
  Y =   \big(y_{1},\ldots, y_{N} \big) \,, \qquad \, y_{1}\geq \ldots \geq y_{N} \geq 0 \,.
\end{equation}
Denote by 
$    Y_{t} = \Big\{y_i \in Y\,|\, y_i \geq t\Big\}$
the set of elements from $Y$ that are greater or equal than $t$.
Then there exists a $0<t\leq y_1$, such that 
\begin{equation}\label{eq:lemma_ineq}
   t\cdot \big|Y_{t}\big| \quad\geq\quad  \frac{1}{h_{N}} \cdot\,\sum_{i=1}^{N}y_{i} \,.
\end{equation}
\end{lemma}
\begin{proof}
Consider
\begin{equation}
    M = \max_{1 \leq \ell \leq N}\quad \ell\, y_{\ell} \,.
\end{equation} 
Denote by $\ell^{*}$ the index such that $M=\ell^{*}y_{\ell^{*}}$~\footnote{While $\ell^*$ might not be unique, this will not matter for the proof. Pick any.}.
By definition, for any $1\leq i \leq N$  it holds that 
\begin{equation}\label{eq:lemma_basic_ineq}
    \frac{1}{i}M \geq y_i \,.
\end{equation}
Summing Eq.~\eqref{eq:lemma_basic_ineq} over the values of $Y$
and using the partial harmonic series [Eq.~\eqref{eq:part_harmonic}],
\begin{equation}
    h_{N}\, M \geq \sum_{i=1}^{N}y_{i}\,.
\end{equation}
Now set $t = y_{\ell^{*}}$.
Since the sequence $Y$ is non-increasing, 
the number of its elements that are greater or equal $y_{\ell^{*}}$ is at least $\ell^{*}$, that is, $|Y_{t}| \geq \ell^{*}$.
With $M = \ell^{*}y_{\ell^{*}}$ we have
\begin{equation}
    t \, |Y_{t}|  \geq \ell^{*}y_{\ell^{*}}
    \geq \frac{1}{h_{N}} \sum_{i=1}^{N}y_{i} \,.
\end{equation}

Last we show that $t\leq y_1$. To see this, assume otherwise $t > y_1$. Then $\big|Y_{t}\big| = 0$ and Ineq.~\eqref{eq:lemma_ineq} cannot hold, leading to a contradiction. Thus $t\leq y_1$.

This ends the proof.
\end{proof}
Clearly, any sequence of non-negative numbers can be sorted so that 
Lemma~\ref{lem:threshold-selection} applies.

\subsection{Proof of Theorem~\ref{thm:general-amplification}}
\label{subs:beta_theorem_proof}

We proceed with the proof.

\begin{proof}
Given a graph $G$, let $\AAA$ be its Pauli realization.
Then by the construction in Section~\ref{sec:lex_prod_obs},
$\AAA^m = \AAA^{\lex m}$ is a Pauli realization of $G^{\lex m}$.
Recall that if $\AAA$ has $N$ elements, then $\AAA^m$ contains $N^m$ elements.

Now let $\sigma$ be the state that achieves the value $\beta(G)$.
Likewise, let $\sigma_m$ be the state that 
achieves the value $\beta(G^{\lex m}) = \beta(G)^m$.
Here we have used that the $\beta$ number is multiplicative under the lexicographic graph product [Eq.~\eqref{eq:lex_product_multi}].
Note that given $\sigma$
it is not obvious how to construct $\sigma_m$ explicitly, but this will not matter for the proof.

From now on, fix $m$ and consider the sequence 
\begin{equation}
    y_A = 
    |\tr(A \sigma_m)|^2\,,\quad  A \in \AAA^m\,,
\end{equation}
Order the set $\{y_A\}_{A \in \AAA^m}$ non-increasingly, 
\begin{equation}
    Y = \big( y_1, \dots, y_{N^m} \big)\,, \quad \quad  y_1 \geq \dots \geq y_{N^m} \geq 0\,.
\end{equation}
Note that $Y$ contains at least one non-zero element 
due to the fact that $\sigma_m$ is the maximizer 
of $\beta(G^{\lex m})$, it can only vanish for empty graphs.

We can then invoke Lemma~\ref{lem:threshold-selection}: 
For this, recall that 
$    Y_{t} = \Big\{y_i \in Y\,|\, y_i \geq t\Big\}$
is the set of elements from $Y$ that are greater or equal $t$.
The Lemma states that there exists $0< t \leq y_1$, 
such that
\begin{equation}
    t\cdot \big|Y_{t}\big| \quad\geq\quad  \frac{1}{h_{N^m}} \cdot\,\sum_{i=1}^{N^m}y_{i} \,.
\end{equation}
For given $m$, denote it as $t_m$.
Rewriting the right hand side in terms of expectation values
and using the fact that all $y_A$ arise from measuring observables of $\AAA^m$ on $\sigma_m$, 
this reads
\begin{equation}\label{eq:threshold_A}
    t_m\cdot \big|Y_{t_m}\big| \quad\geq\quad  
    \frac{1}{h_{N^m}} \cdot\,\beta(G^{\lex m}) \,.
\end{equation}
Now set $\epsilon_m =  \sqrt{0.99 \,t_{m}}$. 
Then $\epsilon_m^{2} < t_m \leq \max_{A \in \AAA^m} 
y_A \leq 1$, as all $A$ are Pauli strings.
This guarantees $\epsilon_{m}< 1$, 
thereby satisfying  the requirement in Conjecture~\ref{conj:13}.

Consider a set of observables from $\AAA^{m}$ and some threshold $\delta$. Define
\begin{equation}
    \AAA^{m}_{\delta} = \{ A \in \AAA^{m} \,:\, |\tr(A\sigma_{m})| \geq \delta \}\,.
\end{equation}
Recall the analogously defined quantity
\begin{equation}
    B_{\delta} = 
    \big\{
    A \in \Pauli \,:\, 
    |\tr(A \sigma_m)| \geq \delta
    \big\}\,,
\end{equation}
where $\Pauli$ is the set of all Pauli strings 
on the space of $\sigma_m$. 

Let us observe two relevant graph inclusions:
First, by $\epsilon_m \leq \sqrt{t_m}$
and $\AAA^m \subseteq \Pauli$,
one gets the graphs inclusions,
\begin{align}\label{eq:sets_inclusion_sequence}
    \AAA^m_{ \sqrt{t_m}} 
    \quad\subseteq\quad
    \AAA^{m}_{ \epsilon_m} 
    \quad\subseteq\quad
    B_{\epsilon_m}\,.
\end{align}
For $H$ an induced subgraph of $G$ write 
$H \subseteq_i G$.
From Eq.~\eqref{eq:sets_inclusion_sequence} then follows
\begin{equation}
\label{eq:graph_inlcusion1}
G(\AAA^m_{\sqrt{t_m}}) \quad \subseteq_i \quad 
G(\AAA_{\epsilon_m}^m)    \quad \subseteq_i \quad 
G(B_{\epsilon_m})    \,.
\end{equation}
Second, from $\AAA^m_{\sqrt{t_m}} \subseteq \AAA^m$ follows
\begin{equation}
\label{eq:graph_inclusion2}
    G(\AAA^{m}_{\sqrt{ t_{m}}})\quad \subseteq_i \quad
    G(\AAA^m) \quad = \quad G^{\lex m}\,.
\end{equation}

Recalling that $Y$ is the set of squared expectation values 
$|\tr(A \sigma_m)|^2$,

\begin{equation} \label{eq:sets_sequence_card}
    \big|Y_{t_m}\big| = |\AAA^m_{\sqrt{ t_m}}\big|
    = 
    \big|V(\AAA^m_{\sqrt {t_m}})\big|\,.
\end{equation}

Let us restate the necessary graph facts from Section~\ref{sec:chromatic_number}:
\begin{enumerate}[label=(\alph*)]\setlength{\itemsep}{2pt}
\item 
Let $H$ be an induced
subgraph of $G$. 
Then 
$\alpha(H)\leq\alpha(G)$
and
$\chif(H)\leq\chif(G)$.

\item 
For all $G$, $H$: $\alpha(G \lex H) = \alpha(G) \cdot \alpha(H)$ 
and $\beta(G \lex H) = \beta(G) \cdot \beta(H)$.

\item For all $G$, $\chif(G) \geq \frac{|V(G)|}{\alpha(G)}$.

\end{enumerate}

We are ready for the final chain of inequalities.
\begin{align}
    \label{eq:beta_sequence}
     \chif(B_{\epsilon_{m}}) \cdot \epsilon_{m}^{2} 
     \quad &\geq\quad  
     \chif(\AAA^{m}_{\sqrt{ t_m}}) \cdot  \epsilon_{m}^{2} 
     && \text{[Fact (a) for $\chif$ and Eq.~\eqref{eq:graph_inlcusion1}]}\nonumber \\
     &\geq \quad
     \frac{|V(\AAA^{m}_{\sqrt{ t_m}})|}{\alpha(\AAA^{m}_{\sqrt{t_m}})} \cdot (0.99 t_m) && \text{[Fact (c)]} \nonumber\\
     &= \quad
     0.99 \cdot 
     \frac{|Y_{t_{m}}|\cdot t_m }{\alpha(\AAA^{m}_{\sqrt{t_m}})} && \text{[Eq.~\eqref{eq:sets_sequence_card}]} \nonumber\\
     &\geq \quad 
     \frac{ 0.99}{h_{N^m}} \cdot
     \frac{\beta(G^{\lex m})}{\alpha(\AAA^m_{\sqrt{ t_m}})} 
     &&\text{[Eq.~\eqref{eq:threshold_A}]} \nonumber\\
     &\geq \quad 
     \frac{ 0.99}{h_{N^m}} \cdot
     \frac{\beta(G^{\lex m})}{\alpha(G^{\lex m})}
     &&\text{[Fact (a) for $\alpha$ and Eq.~\eqref{eq:graph_inclusion2}]} \nonumber\\
     &= \quad 
     \frac{ 0.99}{h_{N^m}} \cdot
     \bigg( \frac{\beta(G)}{\alpha(G)} \bigg)^m\,.
     &&\text{[Fact (b) for $\alpha$, $\beta$]}
\end{align}

 A standard integral bound on the partial sum of a harmonic sequence is
 (see e.g. the derivation of the Integral Test in~\cite[Theorem 9 in Section~10.3]{weir_thomas_2014})
    \begin{equation}
        h_{N^{m}} = \sum_{k=1}^{N^m} \frac{1}{k} 
        \leq 1+ \int_{1}^{N^{m}}\frac{dx}{x} =  1+ m\log N\,.
    \end{equation}
    One can then lower bound the expression on the left hand side of Eq.~\eqref{eq:beta_sequence} by:
    \begin{equation}
\label{eq:bound}
    \chif(B_{\epsilon_{m}})\cdot \epsilon_{m}^{2} \quad \geq \quad 
    \frac{0.99}{1+m\log N} 
    \bigg( \frac{\beta(G)}{\alpha(G)} \bigg)^{m}\,.
    \end{equation}

This ends the proof of Theorem~\ref{thm:general-amplification}. 

\end{proof}

\section{Conclusion}

We have shown that any graph with $\alpha < \beta$
provides a counterexample to Conjecture~\ref{conj:13}.
The key idea is that $(\tfrac{\beta}{\alpha})$
can be amplified under the lexicographic graph product so that
$\chif(B_\epsilon) \cdot \epsilon^2$ is 
unbounded.

We note that, while Conjecture~\ref{conj:13} does not hold, 
it does not rule out the existence of a triply efficient shadow tomography protocol for arbitrary Pauli sets,
for example a sample-efficient protocol with $\chif$ scaling as 
$\operatorname{poly}(1/\epsilon)$.
In particular, Conjecture~\ref{conj:13} 
is stronger than required for sample-efficient tomography alone: 
{\em any} upper bound on $\chif$ of the form $O(\epsilon^{-\kappa})$ yields a sample-efficient two-copy Clifford measurement protocol
for any set of Paulis. Thus, a natural question is:

\begin{question}
\label{conj:13p}
Let $\varrho$ be an $n$-qubit state, $\epsilon \in (0, 1)$, and let $B_\epsilon \subseteq \Pauli_n$ be the set of all Paulis $P$ such that $| \tr(\varrho P )| \geq \epsilon$. 
Does there exist 
a number $\kappa > 0$,
so that 
every  
anticommutation graph $G(B_\epsilon)$ 
has 
a
fractional coloring 
of size $O(\epsilon^{-\kappa})$?
\end{question}

We leave this question open.

\appendix
\section{An explicit seed eigenvector}\label{app:seed}
Let us recall that the state which was used to construct a counterexample to the Conjecture~\ref{conj:13} in Section~\ref{Sec:proof_main}, was chosen as an eigenstate of the Hamiltonian $H = \sum_{A\in \AAA} A$\,,
where $\AAA$ is a set of Pauli observables from Eq.~\eqref{eq:pauli_set1}.  By direct calculation, 
one can verify that the smallest eigenvalue of $H$ equals
\begin{equation}
    \lambda_{\min} = -(1+2\sqrt{2})\,.
\end{equation}
The corresponding ground state
is 
 \begin{equation}\label{eq:seed-vector}
|\psi_{\text{min}}\rangle
=\frac{1}{2\sqrt{3+\sqrt{2}}}\begin{bmatrix}
\frac{1}{\sqrt2}(-1+i)\\
-1\\
\tfrac{1}{\sqrt2}(1+i)\\
-i\\
1+\sqrt2\\
\tfrac{1}{\sqrt2}(1+\sqrt2)(1-i)\\
-1\\
\tfrac{1}{\sqrt2} (-1+i)
\end{bmatrix}\,.
\end{equation}

\bibliographystyle{unsrt}
\bibliography{references}

\end{document}